\documentclass[a4paper,11pt]{article}
\usepackage[utf8]{inputenc}
\usepackage[T1]{fontenc}
\usepackage{lmodern}
\usepackage{microtype}
\usepackage[a4paper,textwidth=16cm,textheight=22cm,centering]{geometry}

\usepackage{mathtools}
\usepackage{amssymb,amsfonts}
\usepackage{amsthm}
\usepackage{setspace}
\usepackage{indentfirst}
\usepackage{enumitem}
\usepackage{float}
\usepackage{algorithm}

\usepackage{cite}
\usepackage{xcolor}
\usepackage[ ]{hyperref}
\hypersetup{
	colorlinks=true,
	linktoc=all,
	linkcolor=red,
	citecolor=blue,
	urlcolor=black,}

\theoremstyle{plain}
\newtheorem{theorem}{Theorem}[section]
\newtheorem{lemma}[theorem]{Lemma}

\newtheorem{corollary}[theorem]{Corollary}

\newtheorem{conjecture}[theorem]{Conjecture}

\theoremstyle{definition}
\newtheorem{definition}[theorem]{Definition}
\newtheorem{problem}[theorem]{Problem}

\theoremstyle{remark}

\begin{document}
	\date{}
	\begin{spacing}{1.03}

\title{Polynomial Algorithms for Minimum Degree Partitions in Semicomplete Digraphs}

\author{
Hanzhi Bai$^{1}$\and
Jin Yan$^{1}$\thanks{Corresponding author. E-mail: {\tt yanj@sdu.edu.cn.} (J.Yan)
}}

\footnotetext[1]{{School of Mathematics, Shandong University, Jinan 250100, P.~R.~China.
Emails: {\tt bhz@mail.sdu.edu.cn} (H. Bai),
Supported by the National Natural Science Foundation of China (Grant No.~12571373) and Natural Science Foundation of Shandong Province (Grant No. ZR2025MS05)}}
\maketitle

\begin{abstract}
A 2-partition of a digraph is a partition of its vertex set into two nonempty parts. Degree-constrained 2-partition problems are generally computationally difficult, even when the prescribed properties are expressed only in terms of minimum indegree, minimum outdegree, or minimum semidegree.  Bang-Jensen and Christiansen~\cite{B-C} conjectured that the minimum-degree partition problems would be polynomial-time solvable on semicomplete digraphs when the degree thresholds are fixed, and Bang-Jensen and Gutin~\cite{B-G-Classes} posed the related Problems~2.8.15 and~2.8.16.

We resolve this conjecture. More precisely, for every fixed pair of integers $k_1,k_2\ge 2$, we give deterministic polynomial-time algorithms that decide whether a given semicomplete digraph admits a $(\delta^+\geq k_1,\delta^-\geq k_2)$-partition, a $(\delta^+\geq k_1,\delta^0\geq k_2)$-partition, or a $(\delta^0\geq k_1,\delta^0\geq k_2)$-partition, and construct such a partition whenever one exists. Here, $\delta ^+,\delta ^-,\delta ^0$ represent the minimum out-, in-, semi-degree, respectively. The algorithms use small degree certificates, minimal cores, closure and protective-set arguments, and deterministic universal colorings with monotone recoloring, which develop a new method in partition algorithm construction.
\end{abstract}



\section{Introduction}




A \emph{2-partition} of a graph or digraph \(G\) is a partition of its
vertex set into two disjoint nonempty subsets \(V_1\) and \(V_2\); that is, \(V(G)=V_1\cup V_2\).
Let \(P_1\) and \(P_2\) be two prescribed properties of graphs or
digraphs. We say that \(G\) admits a \((P_1,P_2)\)-partition if there
exists a 2-partition \((V_1,V_2)\) of \(V(G)\) such that \(G[V_1]\)
satisfies \(P_1\) and \(G[V_2]\) satisfies \(P_2\). Typically, \(P_1\)
and \(P_2\) may involve connectivity requirements, degree constraints,
the existence or exclusion of certain subgraphs, or other structural
conditions.

The study of graph 2-partitions is an important topic in graph theory
and has attracted considerable attention. Degree constraints have long played a central role in the study of graph partitions, scholars have made contributions in this field. The most classic result is the paper by Stiebitz in 1996~\cite{Stiebitz-1996}. He proved that for positive integers $s$ and $t$, every graph $G$ with the minimum degree $\delta(G)\ge s+t+1$ admits a $(\delta\ge s,\delta\ge t)$-partition. Recently, Wang and Wu~\cite{W-W} proved the average-degree partition theorem conjectured by Cs\'oka, Lo, Norin, Wu, and Yepremyan, while Ma and Yu~\cite{M-Y-2016} settled a conjecture of Bollob\'as and Scott on sparse bipartitions. Further results on graph 2-partitions can be found in~\cite{J-M-Y-Y,M-Y-2019,F-M,L-X,K-O-Z,W-W-L,L-S-S-Z}.


For digraphs, research on degree constrained partition problems is also meaningful. For a digraph $D$, we denote that $\delta^+(D)=\min_{v\in V(D)}d_H^+(v)$, $\delta^-(D)=\min_{v\in V(D)}d_H^-(v)$ and $\delta^0(D)=\min\{\delta^+(D),\delta^-(D)\}$. In~\cite{Alon1996} by Alon in 1996 and~\cite{Stiebitz-1995} by Stiebitz in 1995, they ask, separately, whether for every pair of positive integers \(k_1,k_2\), there is a function \(f(k_1,k_2)\) such that every digraph with minimum outdegree at least \(f(k_1,k_2)\) admits a \((\delta^+\ge k_1,\delta^+\ge k_2)\)-partition. In fact, as early as 1983, Thomassen~\cite{Thomassen-1983} proved that every digraph $D$ with $\delta^+ (D)\ge 3$ yields a $(\delta^+ \ge 1,\delta^+ \ge 1)$-partition, which means that $f(1,1)=3$. But for larger $k_1,k_2$, even the existence of $f(1, 2)$ is still open. Recently, Steiner et al.\cite{ChristophPetrovaSteiner2025} proved that if $f(2,2)$ exists, then all the numbers $f(k_1,k_2)$ with $k_1,k_2 \ge 1$ exist, which has greatly advanced research on this conjecture.

However, from an algorithmic perspective, this conjecture holds true for a special class of digraphs: semicomplete digraphs. Here, a digraph \(S\) is said to be \emph{semicomplete} if, for every pair of distinct vertices \(u,v\in V(S)\), at least one of \(uv\) and \(vu\) is an arc of \(S\). In 2018, Bang-Jensen and Christiansen~\cite{B-C} proposed a polynomial-time algorithm deciding that for a given semicomplete digraph \(S\), whether $S$ has a $(\delta^+\ge k_1,\delta^+\ge k_2)$-partition for $k_1,k_2\ge 1$. Moreover, they generalized this conclusion to a more general case. They proved that each of the following problems can be decided in polynomial time:
\begin{itemize}
    \item whether $S$ has a $(\delta^+\ge 1,\delta^-\ge k)$-partition for $k\ge 1$;
    \item whether $S$ has a $(\delta^+\ge 1,\delta^0\ge 1)$-partition;
    \item whether $S$ has a $(\delta^0\ge 1,\delta^0\ge 1)$-partition.
\end{itemize}

Since these results do not resolve the corresponding partition problems
under arbitrary minimum-degree requirements, Bang-Jensen and
Christiansen proposed the following conjecture:
\begin{conjecture}\label{conj1}
For every fixed pair of integers $k_1,k_2\ge 2$, each of the following problems can be decided in polynomial time for a given semicomplete digraph $S$:
\begin{itemize}
    \item whether $S$ has a $(\delta^+\ge k_1,\delta^-\ge k_2)$-partition;
    \item whether $S$ has a $(\delta^+\ge k_1,\delta^0\ge k_2)$-partition;
    \item whether $S$ has a $(\delta^0\ge k_1,\delta^0\ge k_2)$-partition.
\end{itemize}
\end{conjecture}

In particular, Bang-Jensen and Gutin also raise the same problem in~\cite{B-G-Classes}, Problems~2.8.15 and~2.8.16. Conjecture~\ref{conj1} consists precisely of the symmetry-inequivalent cases left open in~\cite{B-C}: taking converses swaps $\delta^+$ and $\delta^-$, while exchanging parts swaps the two conditions. Together with the known $(\delta^+\ge k_1,\delta^+\ge k_2)$ case, it completes the polynomial-time classification for fixed $k_1,k_2\ge2$.

We resolve Conjecture~\ref{conj1} as our main result in the following theorem.

\begin{theorem}\label{main theorem}
For every fixed pair of integers $k_1,k_2\ge 2$, each of the following problems can be decided in polynomial time for a given semicomplete digraph $S$:
\begin{itemize}
    \item whether $S$ has a $(\delta^+\ge k_1,\delta^-\ge k_2)$-partition;
    \item whether $S$ has a $(\delta^+\ge k_1,\delta^0\ge k_2)$-partition;
    \item whether $S$ has a $(\delta^0\ge k_1,\delta^0\ge k_2)$-partition.
\end{itemize}
In each case, a required partition can also be constructed in polynomial time whenever one exists.
\end{theorem}

The main novelty of this paper is that it avoids a direct search over all partitions. Instead, it uses the special structure of semicomplete digraphs to identify small guiding sets from which the desired partition can be recovered. In the third case, the paper further develops an algebraic construction of a universal coloring family, which replaces random guessing by a deterministic set of colorings guaranteed to contain a useful initial coloring. Combined with a monotone recoloring procedure, this yields constructive polynomial-time algorithms for all the remaining fixed-threshold cases.

Combining Theorem~\ref{main theorem} with the earlier $(\delta^+,\delta^+)$ result, gives the following uniform consequence and makes the scope of the main result explicit.

\begin{corollary}\label{cor:all-degree-types}
Let $k_1,k_2\ge 2$ be fixed, and let $\sigma_1,\sigma_2\in\{+,-,0\}$. There is a polynomial-time algorithm that, given a semicomplete digraph $S$, decides whether $S$ has a $(\delta^{\sigma_1}\ge k_1,\delta^{\sigma_2}\ge k_2)$-partition and constructs such a partition whenever one exists.
\end{corollary}

Consequently, Theorem~\ref{main theorem} completes the polynomial-time classification of all nine ordered combinations of minimum outdegree, minimum indegree, and minimum semidegree constraints on the two parts when the two thresholds are fixed integers at least two. Since every tournament is semicomplete, the same conclusion holds in particular for tournaments.

\subsection*{Organization of the paper}

In Section~\ref{section 2}, we collect the notation and preliminary definitions. In Section~\ref{section 3}, we prove the main result. We use Theorems~\ref{thm:main},~\ref{thm:semi-out-degree-partition}, and~\ref{thm:fixed-semidegree-partition} to prove the three parts of Theorem~\ref{main theorem}. In Section~\ref{section 4}, we give concluding remarks and pose an open problem. Finally, in Appendix~\ref{app:universal-coloring-calculations}, we provide the detailed calculations used in the proof of Lemma~\ref{lem:universal-coloring}.

\section{Preliminaries}\label{section 2}

All digraphs in this paper are finite and loopless. For a digraph
$D=(V(D),A(D))$, its order is $n=|V(D)|$. For $X\subseteq V(D)$, the
induced subdigraph $D[X]$ has vertex set $X$ and arc set
$A(D)\cap(X\times X)$; the converse digraph $\overleftarrow D$ is obtained
by reversing every arc. The digraph $D$ is \emph{semicomplete} if, for
every pair of distinct vertices $u,v$, at least one of $uv$ and $vu$ is an
arc; both arcs are allowed. We use
\(
 N_D^+(v)=\{u:vu\in A(D)\},~
 N_D^-(v)=\{u:uv\in A(D)\},
\)
and, for $X\subseteq V(D)$, write $N_X^\pm(v)=N_D^\pm(v)\cap X$ and
$d_X^\pm(v)=|N_X^\pm(v)|$. For a nonempty digraph $H$, let
$\delta^\pm(H)=\min_{v\in V(H)}d_H^\pm(v)$ and
$\delta^0(H)=\min\{\delta^+(H),\delta^-(H)\}$; the latter is the
\emph{minimum semidegree}. An \emph{ordered $2$-partition} of $D$ is a pair $(V_1,V_2)$ of disjoint
nonempty sets with $V_1\cup V_2=V(D)$. For $U\subseteq V(D)$, let $K_a(U)$ be the unique
maximum set $S\subseteq U$ satisfying $\delta^+(D[S])\ge a$, and put
$K_a(U)=\varnothing$ if no such nonempty set exists. For
$\sigma,\tau\in\{+,-,0\}$, it is a
$(\delta^\sigma\!\ge a,\delta^\tau\!\ge b)$-partition if
$\delta^\sigma(D[V_1])\ge a$ and $\delta^\tau(D[V_2])\ge b$. These conventions follow standard digraph
terminology~\cite{B-G}. We
write $[n]=\{1,\ldots,n\}$, all logarithms are to base $2$, and
$\lceil x\rceil$ denotes the ceiling of $x$.

Next, we provide the definitions of \emph{$b$-core} and \emph{closure}:
\begin{definition}[Vertex-minimal indegree $b$-core]\label{def:vertex-minimal-indegree-core}
A nonempty set $B\subseteq V(D)$ is a \emph{vertex-minimal indegree
$b$-core} if $\delta^-(D[B])\ge b$ but
$\delta^-(D[B\setminus\{x\}])<b$ for every $x\in B$.
\end{definition}

\begin{definition}[The closure $\operatorname{cl}_{b,E}(B)$]\label{def:indegree-closure}
If $E\cap B=\varnothing$ and $\delta^-(D[B])\ge b$, set $S_0=B$ and, for
$j\ge0$, define
\(S_{j+1}=S_j\cup\{v\in V(D)\setminus(E\cup S_j):d_{S_j}^-(v)\ge b\}.\)
The increasing sequence stabilizes after at most $\lvert V(D)\rvert$ steps;
its stable value is denoted by $\operatorname{cl}_{b,E}(B)$.
\end{definition}

We finally record the algebraic and probabilistic notation used in
Lemma~\ref{lem:universal-coloring} and Appendix~\ref{app:universal-coloring-calculations}. Let $\mathbb F_2=\{0,1\}$ and
$\mathbb F=\operatorname{GF}(2^r)$, viewed as an $r$-dimensional vector
space over $\mathbb F_2$. A monic polynomial has leading coefficient $1$,
and it is irreducible if it has no nonconstant proper factor; a monic
irreducible polynomial $f$ of degree $r$ yields the representation
$\mathbb F\cong\mathbb F_2[Z]/(f)$. We use the facts that multiplication
by a nonzero field element is bijective, every nonzero
$\mathbb F_2$-linear functional $\lambda:\mathbb F\to\mathbb F_2$ has
kernel size $2^{r-1}$, and a nonzero polynomial of degree $d$ has at most
$d$ roots~\cite{L-N}. The number-theoretic M\"obius
function satisfies $\mu(1)=1$, $\mu(m)=(-1)^s$ when $m$ is a product of
$s$ distinct primes, and $\mu(m)=0$ otherwise. The vector $e_i$ is the
$i$th standard basis vector of $\mathbb F_2^n$.

An \emph{indexed multiset} $(x_\alpha)_{\alpha\in I}$ may contain equal
objects under different indices; uniform sampling is over $I$, so
repetitions retain their multiplicities. An indexed family
$\mathcal U_{n,t}=(\chi_\alpha)_{\alpha\in I}$ of maps
$\chi_\alpha:[n]\to\{0,1\}$ is \emph{$t$-universal} if every map
$\eta:T\to\{0,1\}$ with $|T|\le t$ equals $\chi_\alpha|_T$ for some
$\alpha$. We write $\mathbf 1_{\mathcal E}$ for the indicator of an event,
and $\mathbb E$ and $\Pr$ for expectation and probability over the stated
indexed multiset. For $R\subseteq[n]$, the Walsh character is
$\psi_R(\chi)=(-1)^{\sum_{i\in R}\chi(i)}$; a family is
$\varepsilon$-biased if $|\mathbb E\psi_R|\le\varepsilon$ for every
nonempty $R$. The identity
\[
 \mathbf 1_{\{\chi|_T=\eta\}}
 =2^{-|T|}\sum_{R\subseteq T}
 (-1)^{\sum_{i\in R}\eta(i)}\psi_R(\chi)
\]
is the Walsh--Fourier expansion used in Appendix~\ref{app:universal-coloring-calculations}; see
\cite{O} for background.

\section{Algorithmic Constructions and Correctness Proofs}\label{section 3}

\subsection{Algorithm for $(\delta^+\ge k_1,\delta^-\ge k_2)$-partition}

\begin{theorem}
\label{thm:main}
For every fixed pair of positive integers $k_1,k_2$, Algorithm 1 described
below correctly decides whether a semicomplete digraph $D$ has a
$(\delta^+\ge k_1,\delta^-\ge k_2)$-partition and constructs such a
partition whenever one exists. Moreover, for $n=\lvert V(D)\rvert$, its
running time is $O\bigl(n^{k_1(2k_1+1)+k_2(2k_2+1)+2}\bigr)$.
\end{theorem}

We first establish two boundedness lemmas used by the algorithm.

\begin{lemma}
\label{lem:few-low-degree}
Let $H$ be a semicomplete digraph and let $r\ge 0$ be an integer. Then at most $2r+1$
vertices of $H$ have outdegree at most $r$, and at most $2r+1$ vertices of $H$
have indegree at most $r$.
\end{lemma}

\begin{proof}
It suffices to prove the outdegree statement, since the indegree statement
follows by applying the same argument to the converse digraph. Let
$S=\{v\in V(H):d_H^+(v)\le r\}$ and write $s=\lvert S\rvert$. By
semicompleteness, every unordered pair of distinct vertices in $S$ contributes
at least one arc to $H[S]$. Hence the sum of the outdegrees in $H[S]$ is at
least $\binom{s}{2}$, whereas it is at most $rs$ because every vertex of $S$
has outdegree at most $r$ in $H$. Thus $\binom{s}{2}\le rs$. If $s>0$, then
$s-1\le 2r$, and therefore $s\le 2r+1$; the case $s=0$ is immediate.
\end{proof}

For $b\ge 1$, every vertex-minimal indegree $b$-core has at least $b+1$
vertices. Indeed, for such a core $B$ and every $v\in B$, we have
$b\le d_B^-(v)\le\lvert B\rvert-1$ because the digraph is loopless.

\begin{lemma}
\label{lem:minimal-core}
Let $b\ge 1$, and let $B$ be a vertex-minimal indegree $b$-core in a
semicomplete digraph $D$. Then $\lvert B\rvert\le b(2b+1)$.
\end{lemma}

\begin{proof}
Let $C=\{v\in B:d_B^-(v)=b\}$. Every vertex of $C$ has indegree at most $b$ in
$D[C]$, so Lemma~\ref{lem:few-low-degree} gives $\lvert C\rvert\le 2b+1$.

Fix $x\in B$. By the vertex-minimality of $B$, some
$y\in B\setminus\{x\}$ has indegree less than $b$ in $D[B\setminus\{x\}]$.
Since $d_B^-(y)\ge b$ and deleting one vertex decreases an indegree by at most
one, we must have $d_B^-(y)=b$ and $xy\in A(D)$. Thus $y\in C$ and
$x\in N_B^-(y)$. As $x$ was arbitrary, $B$ is covered by the sets
$N_B^-(y)$ with $y\in C$. Each such set has exactly $b$ vertices, and hence
$\lvert B\rvert\le b\lvert C\rvert\le b(2b+1)$.
\end{proof}

\setcounter{theorem}{3}

For $b\ge 1$, let $E,B\subseteq V(D)$ be disjoint with
$\delta^-(D[B])\ge b$, and consider
$\operatorname{cl}_{b,E}(B)$ as defined in Section~\ref{section 2}. Its defining
process stabilizes after at most $\lvert V(D)\rvert$ rounds. The initial vertices
already satisfy the required indegree condition, and every added vertex has at
least $b$ in-neighbors among vertices added earlier. Consequently
$\delta^-(D[\operatorname{cl}_{b,E}(B)])\ge b$.

The closure has the following absorption property. Suppose that distinct
vertices $w_1,\ldots,w_t\in V(D)\setminus E$ can be ordered so that each $w_i$
has at least $b$ in-neighbors in $B\cup\{w_1,\ldots,w_{i-1}\}$. Then all these
vertices belong to $\operatorname{cl}_{b,E}(B)$. Indeed, once $w_1,\ldots,w_{i-1}$ have been
absorbed, the closure contains the set that supplies the required
in-neighbors of $w_i$, so stability forces $w_i$ to be absorbed as well.

\medskip

\begin{algorithm}[H]
\caption{Deciding a $(\delta^+\ge k_1,\delta^-\ge k_2)$-partition}
\label{alg:plus-minus-partition}
\noindent\textbf{Input:} A semicomplete digraph $D$ given by its adjacency matrix, and fixed positive integers $k_1,k_2$.\par
\noindent\textbf{Output:} A $(\delta^+\ge k_1,\delta^-\ge k_2)$-partition of $D$, or a declaration that no such partition exists.
\begin{enumerate}[label=\textup{(\arabic*)},leftmargin=*,itemsep=2pt]
    \item Put $a:=k_1$, $b:=k_2$, $p:=a(2a+1)$, and $q:=b(2b+1)$; the bounds $p$ and $q$ are those supplied by Lemmas~\ref{lem:few-low-degree} and~\ref{lem:minimal-core}.
    \item Enumerate every set $E\subseteq V(D)$ with $\lvert E\rvert\le p$.
    \item For each such $E$, enumerate every set $B\subseteq V(D)\setminus E$ satisfying $b+1\le\lvert B\rvert\le q$ and $\delta^-(D[B])\ge b$.
    \item Compute $Y:=\operatorname{cl}_{b,E}(B)$ as in Definition~\ref{def:indegree-closure}, and let $X:=V(D)\setminus Y$.
    \item If $X$ is nonempty and $\delta^+(D[X])\ge a$, output $(X,Y)$.
    \item If all candidate pairs are exhausted without an output, answer that no required partition exists.
\end{enumerate}
\end{algorithm}

\medskip

\begin{proof}[Proof of Theorem~\ref{thm:main}]
\emph{Soundness.}
Suppose that the algorithm outputs $Y=\operatorname{cl}_{b,E}(B)$ and
$X=V(D)\setminus Y$. It directly verifies that $X$ is nonempty and that
$\delta^+(D[X])\ge a$. Moreover, $B\subseteq Y$ and $\lvert B\rvert\ge b+1$,
so $Y$ is nonempty, while the construction of the closure gives
$\delta^-(D[Y])\ge b$. Thus $(X,Y)$ is a valid
$(\delta^+\ge a,\delta^-\ge b)$-partition.

\emph{Completeness.}
Suppose that $D$ has such a partition. Among all feasible partitions
$V(D)=X\cup Y$ with $\delta^+(D[X])\ge a$ and
$\delta^-(D[Y])\ge b$, choose one with $\lvert X\rvert$ minimum. Let
$C=\{u\in X:d_X^+(u)=a\}$ and let $E$ be the union of the out-neighborhoods
$N_{D[X]}^+(u)$ over all $u\in C$. Every vertex of $C$ has outdegree at most
$a$ in $D[C]$, so Lemma~\ref{lem:few-low-degree} gives
$\lvert C\rvert\le 2a+1$. Since each $u\in C$ has exactly $a$ out-neighbors
in $X$, it follows that $\lvert E\rvert\le a\lvert C\rvert\le p$.

The minimal choice of $X$ implies the restriction that
\(d_Y^-(x)<b\) for every \(x\in X\setminus E\). To see this, fix $x\in X\setminus E$. Deleting $x$ from $X$ preserves minimum outdegree at least $a$. Indeed, a vertex $u$ that does not send an arc to $x$ loses no out-neighbor, while a vertex $u$ with $ux\in A(D)$ cannot belong to $C$, because $x\notin E$; hence $d_X^+(u)\ge a+1$ before the deletion. If $d_Y^-(x)\ge b$ also held, moving $x$ from $X$ to $Y$ would preserve both degree conditions and produce a feasible partition with a smaller first part, a contradiction.

Starting from $Y$, repeatedly delete a vertex whenever the remaining induced
subdigraph still has minimum indegree at least $b$. The process terminates with
a vertex-minimal indegree $b$-core $B\subseteq Y$, and
Lemma~\ref{lem:minimal-core} gives $b+1\le\lvert B\rvert\le q$. Reverse the
deletion order and write $Y\setminus B=\{w_1,\ldots,w_t\}$. At the moment a
vertex $w_i$ was deleted, it had at least $b$ in-neighbors among the vertices
that remained. Therefore, in the reverse order, each $w_i$ has at least $b$
in-neighbors in $B\cup\{w_1,\ldots,w_{i-1}\}$.

The algorithm enumerates this particular pair $(E,B)$ because $E\subseteq X$,
$B\subseteq Y$, and both sets satisfy the required size bounds. Let
$S=\operatorname{cl}_{b,E}(B)$. The absorption property shows that every vertex of
$Y\setminus B$ belongs to $S$, and hence $Y\subseteq S$. Conversely, the closure
cannot add a vertex of $X\setminus E$: at the first round in which such a
vertex $x$ was added, the current closure set would still be contained in
$Y$, so $x$ would have at least $b$ in-neighbors in $Y$, contradicting the key
restriction above. The closure never adds a vertex of $E$ by definition.
Thus $S\subseteq Y$, and we obtain the central identity
\(
\operatorname{cl}_{b,E}(B)=Y.
\)
When the algorithm reaches $(E,B)$, it therefore reconstructs the feasible
partition exactly and accepts.

\emph{Running time.}
There are $O(n^{p+q})$ candidate pairs $(E,B)$ because $p$ and $q$ are fixed.
For each pair, the closure and the two degree tests can be computed in
$O(n^2)$ time from an adjacency matrix. The total running time is therefore
$O(n^{p+q+2})$, which is polynomial in $n$ for fixed $k_1,k_2$.
\end{proof}

\subsection{Algorithm for $(\delta^+\ge k_1,\delta^0\ge k_2)$-partition}

\begin{theorem}\label{thm:semi-out-degree-partition}
For every fixed pair of positive integers $k_1,k_2$, Algorithm 2 that decides whether a given semicomplete digraph $D$ has a
$(\delta^+\ge k_1,\delta^0\ge k_2)$-partition and constructs such a partition
whenever one exists. Moreover, writing $h=1+k_1(2k_1+1)$, $r=k_2+h$, and
$q=2(2r-1)(r-1)$, its running time is $O(n^{q+2})$, where
$n=\lvert V(D)\rvert$.
\end{theorem}

We first record the threshold-shifted form of
Lemma~\ref{lem:few-low-degree}.

\begin{lemma}[Low-degree threshold shift]
\label{lem:low-degree-threshold-shift}
Let $H$ be a semicomplete digraph and let $r\ge1$ be an integer. Then at most
$2r-1$ vertices of $H$ have outdegree less than $r$, and at most $2r-1$
vertices of $H$ have indegree less than $r$.
\end{lemma}

\begin{proof}
Since degrees are integers, having degree less than $r$ is equivalent to
having degree at most $r-1$. The result follows from
Lemma~\ref{lem:few-low-degree} applied with the integer parameter $r-1$.
\end{proof}

\begin{lemma}[Deleting vertices outside a protective set]
\label{lem:protective-set}
Let $b\ge1$ and $h\ge0$ be integers, let $Y\subseteq V(D)$ satisfies
$\delta^0(D[Y])\ge b$, and put $r=b+h$. Define
$L^+=\{v\in Y:d_Y^+(v)<r\}$ and $L^-=\{v\in Y:d_Y^-(v)<r\}$, and let
\[
Z_Y(h):=\bigcup_{v\in L^+}N_Y^+(v)\;\cup\!\bigcup_{v\in L^-}N_Y^-(v).
\]
Then $\lvert Z_Y(h)\rvert\le 2(2r-1)(r-1)$. Moreover, whenever
$E\subseteq Y\setminus Z_Y(h)$ and $\lvert E\rvert\le h$, the set
$Y\setminus E$ is nonempty and satisfies
$\delta^0(D[Y\setminus E])\ge b$.
\end{lemma}

\begin{proof}
Apply Lemma~\ref{lem:low-degree-threshold-shift} to $D[Y]$. Hence
$\lvert L^+\rvert,\lvert L^-\rvert\le 2r-1$. Every vertex of $L^+$ has at most
$r-1$ out-neighbors in $Y$, and every vertex of $L^-$ has at most $r-1$
in-neighbors there, which gives the asserted bound on $Z_Y(h)$.

Fix $u\in Y\setminus E$. If $d_Y^+(u)<r$, then $u\in L^+$ and all its
out-neighbors in $Y$ lie in $Z_Y(h)$. Since $E$ is disjoint from this
protective set, none of them is deleted, and therefore
$d_{Y\setminus E}^+(u)=d_Y^+(u)\ge b$. If $d_Y^+(u)\ge r$, deleting at most
$h$ vertices leaves $d_{Y\setminus E}^+(u)\ge r-h=b$. The symmetric argument
gives $d_{Y\setminus E}^-(u)\ge b$.

It remains to verify nonemptiness. If $E=Y$, then $Y$ is nonempty and
$\lvert Y\rvert\le h$, so $h\ge1$ and every $v\in Y$ satisfies
$d_Y^+(v)\le\lvert Y\rvert-1\le h-1<r$. Thus $L^+=Y$. Since
$\delta^-(D[Y])\ge b\ge1$, every vertex of $Y$ is an out-neighbor of another
vertex of $Y$, and hence $Y\subseteq Z_Y(h)$. This contradicts
$E\subseteq Y\setminus Z_Y(h)$.
\end{proof}

\begin{lemma}[A small witness through a prescribed vertex]
\label{lem:small-out-witness}
Let $a$ be a positive integer, let $X\subseteq V(D)$ satisfy
$\delta^+(D[X])\ge a$, and let $x\in X$. There
exists a set $C\subseteq X$ containing $x$ such that
$\delta^+(D[C])\ge a$ and $\lvert C\rvert\le h_a:=1+a(2a+1)$.
\end{lemma}

\begin{proof}
Choose an inclusion-minimal subset $C\subseteq X$ that contains $x$ and
satisfies $\delta^+(D[C])\ge a$, and put
$L=\{u\in C:d_C^+(u)=a\}$. For every $y\in C\setminus\{x\}$, the set
$C\setminus\{y\}$ fails the minimum-outdegree condition. Hence some
$u\in C\setminus\{y\}$ has outdegree less than $a$ after $y$ is deleted. Since
$d_C^+(u)\ge a$ and one deletion changes this degree by at most one, we have
$d_C^+(u)=a$ and $uy\in A(D)$. Consequently
$C\setminus\{x\}\subseteq\bigcup_{u\in L}N_C^+(u)$.

Every vertex of $L$ has outdegree at most $a$ in $D[L]$, so
Lemma~\ref{lem:few-low-degree} gives $\lvert L\rvert\le2a+1$. Each vertex of
$L$ has exactly $a$ out-neighbors in $C$, and therefore
$\lvert C\rvert\le1+a\lvert L\rvert\le1+a(2a+1)$.
\end{proof}

\begin{lemma}[Maximum out-core]
\label{lem:maximum-out-core}
Let $a$ be a positive integer and let $U\subseteq V(D)$. The set $K_a(U)$
introduced in Section~\ref{section 2} is well defined and can be computed in $O(\lvert V(D)\rvert^2)$ time.
\end{lemma}

\begin{proof}
The union of two sets satisfying the minimum-outdegree condition satisfies the
same condition: each vertex belongs to one of the original sets and already
has at least $a$ out-neighbors there. Since $D$ is finite, the union of all
qualifying subsets of $U$, when at least one exists, is itself qualifying and
contains every other such subset. This proves existence and uniqueness.

To compute the maximum set, start with $R=U$ and repeatedly delete a vertex
$v$ with $d_R^+(v)<a$. If a set $S\subseteq R$ satisfies
$\delta^+(D[S])\ge a$, then $v\notin S$, because otherwise
$a\le d_S^+(v)\le d_R^+(v)<a$. Thus each deletion preserves every qualifying
subset of $U$. At termination, the remaining set is either empty or has
minimum outdegree at least $a$; in the latter case it contains every other
qualifying subset.

For an adjacency-matrix implementation, compute the current outdegrees in
$R$ and place every vertex of outdegree less than $a$ in a queue. When a vertex
$v$ is deleted, scan the vertices still in $R$ that send an arc to $v$, update
their outdegrees, and enqueue those that newly fall below $a$. Each vertex is
deleted at most once and each deletion requires one linear scan, so the total
time is $O(\lvert V(D)\rvert^2)$.
\end{proof}

\medskip

\begin{algorithm}[H]
\caption{Deciding a $(\delta^+\ge k_1,\delta^0\ge k_2)$-partition}
\label{alg:plus-zero-partition}
\noindent\textbf{Input:} A semicomplete digraph $D$ given by its adjacency matrix, and fixed positive integers $k_1,k_2$.\par
\noindent\textbf{Output:} A $(\delta^+\ge k_1,\delta^0\ge k_2)$-partition of $D$, or a declaration that no such partition exists.
\begin{enumerate}[label=\textup{(\arabic*)},leftmargin=*,itemsep=2pt]
    \item Set $h:=1+k_1(2k_1+1)$, $r:=k_2+h$, and $q:=2(2r-1)(r-1)$; here $h$ is the witness bound from Lemma~\ref{lem:small-out-witness}, while $q$ is the protective-set bound from Lemmas~\ref{lem:low-degree-threshold-shift} and~\ref{lem:protective-set}.
    \item Enumerate every set $Z\subseteq V(D)$ with $\lvert Z\rvert\le q$.
    \item For each such set, compute $P_Z:=K_{k_1}(V(D)\setminus Z)$ using Lemma~\ref{lem:maximum-out-core}, and put $Q_Z:=V(D)\setminus P_Z$.
    \item If both sets are nonempty and $\delta^0(D[Q_Z])\ge k_2$, output $(P_Z,Q_Z)$.
    \item If no candidate passes this test, answer that no required partition exists.
\end{enumerate}
\end{algorithm}

\medskip

\begin{proof}[Proof of Theorem~\ref{thm:semi-out-degree-partition}]
\emph{Soundness.}
Whenever the algorithm outputs $(P_Z,Q_Z)$, the definition of $K_{k_1}$ gives
$\delta^+(D[P_Z])\ge k_1$, and the algorithm directly verifies that
$\delta^0(D[Q_Z])\ge k_2$ and that both parts are nonempty. Hence every output
is valid.

\emph{Completeness.}
Suppose that a required partition exists. Among all ordered partitions
$V(D)=P\cup Q$ satisfying $\delta^+(D[P])\ge k_1$ and
$\delta^0(D[Q])\ge k_2$, choose one with $\lvert P\rvert$ maximum. Apply
Lemma~\ref{lem:protective-set} to $Y=Q$, $b=k_2$, and the value of $h$ fixed by
the algorithm. It produces a set $Z=Z_Q(h)$ with $\lvert Z\rvert\le q$, so the
algorithm enumerates this $Z$. Since $Z\subseteq Q$, we have
$P\subseteq V(D)\setminus Z$. The maximum property established in
Lemma~\ref{lem:maximum-out-core} therefore gives $P\subseteq K$, where
$K=K_{k_1}(V(D)\setminus Z)$.

We claim that $K=P$. Otherwise choose $x\in K\setminus P$. Then
$x\in Q\setminus Z$. By Lemma~\ref{lem:small-out-witness}, the set $K$ contains
a set $C$ with $x\in C$, $\lvert C\rvert\le h$, and
$\delta^+(D[C])\ge k_1$. Put $E=C\cap Q$. The set $E$ is nonempty, is contained
in $Q\setminus Z$, and has size at most $h$. Since both $P$ and $C$ satisfy the
minimum-outdegree condition, so does their union. Since $P\cup Q=V(D)$,
$P\cap Q=\varnothing$, and $C\subseteq V(D)$, we have $C\setminus Q\subseteq P$;
therefore $P\cup C=P\cup(C\cap Q)=P\cup E$. Lemma~\ref{lem:protective-set}
ensures that $Q\setminus E$
is nonempty and satisfies $\delta^0(D[Q\setminus E])\ge k_2$. Thus
$(P\cup E,Q\setminus E)$ is another required partition with a larger first
part, contradicting the choice of $P$. Hence $K=P$.

When the algorithm considers the set $Z$ constructed above, it obtains
$P_Z=P$ and $Q_Z=Q$, and therefore accepts. This proves completeness.

\emph{Running time.}
For $n=\lvert V(D)\rvert$, the number of candidate sets $Z$ is $O(n^q)$ because
$q$ depends only on the fixed parameters. Lemma~\ref{lem:maximum-out-core}
computes each $P_Z$ in $O(n^2)$ time, and the semidegree test has the same
bound. The total running time is $O(n^{q+2})$. Thus the algorithm is polynomial
for fixed $k_1,k_2$; if these integers are part of the input, the same analysis
only gives an XP algorithm parameterized by $(k_1,k_2)$.
\end{proof}

\subsection{Algorithm for $(\delta^0\ge k_1,\delta^0\ge k_2)$-partition}
\label{sec:fixed-semidegree-partitions}

\begin{theorem}\label{thm:fixed-semidegree-partition}
For every fixed pair of integers $k_1,k_2\ge 2$, Algorithm 3 decides whether a given semicomplete digraph
$D$ admits a partition $V(D)=V_1\cup V_2$ such that
\(
\delta^0(D[V_1])\ge k_1 \quad\text{and}\quad \delta^0(D[V_2])\ge k_2,
\)
and constructs such a partition whenever one exists. Moreover, writing
$C=2k_1(2k_1+1)+2k_2(2k_2+1)$, its running time is $O(n^{2C+5})$, where
$n=\lvert V(D)\rvert$.
\end{theorem}

We first establish the auxiliary lemmas used by the algorithm.

\begin{lemma}\label{lem:one-sided-domination}
Let $G$ be a semicomplete digraph on $m\ge 1$ vertices. There is a set
$A\subseteq V(G)$ with $\lvert A\rvert\le \lceil\log(m+1)\rceil$ such that every vertex
$v\in V(G)\setminus A$ has an in-neighbor in $A$. Dually, there is a set $B$
satisfying the same upper bound such that every vertex outside $B$ has an
out-neighbor in $B$. Both sets can be found in polynomial time.
\end{lemma}

\begin{proof}
For every nonempty $U\subseteq V(G)$, the semicompleteness of $G[U]$ gives
$\sum_{u\in U}d_{G[U]}^+(u)\ge \binom{\lvert U\rvert}{2}$. Hence some $x\in U$ satisfies
$d_{G[U]}^+(x)\ge (\lvert U\rvert-1)/2$.

Start with $U_0=V(G)$. Whenever $U_j$ is nonempty, choose such a vertex $x_j$
in $G[U_j]$, add $x_j$ to $A$, and put
$U_{j+1}=U_j\setminus(\{x_j\}\cup N_{G[U_j]}^+(x_j))$. Then
$\lvert U_{j+1}\rvert+1\le (\lvert U_j\rvert+1)/2$, and induction yields
$\lvert U_j\rvert+1\le (m+1)/2^j$. Thus $U_j$ is empty for
$j=\lceil\log(m+1)\rceil$. Every vertex outside $A$ was deleted as an
out-neighbor of a selected vertex, and therefore has an in-neighbor in $A$.
Applying the same construction after reversing all arcs gives $B$. Each step
only requires the computation of outdegrees in an induced subdigraph, so the
construction is polynomial-time.
\end{proof}

\begin{lemma}[Small support set]\label{lem:small-support}
Let $H$ be a semicomplete digraph on $s$ vertices with $\delta^0(H)\ge k$,
where $k\ge 1$, and put $L=\lceil\log(s+1)\rceil$. There is a set
$Q\subseteq V(H)$ that can be constructed in polynomial time such that
$\lvert Q\rvert\le 2k(2k+1)L$ and every $v\in V(H)$ has at least $k$ out-neighbors and
at least $k$ in-neighbors in $Q$.
\end{lemma}

\begin{proof}
We first construct a set $A$ that supplies at least $k$ in-neighbors to every
vertex outside it. Set $A_0=\varnothing$. For $j=1,\ldots,k$, let
$R_j=V(H)\setminus A_{j-1}$. If $R_j=\varnothing$, define $S_j=\varnothing$;
otherwise apply Lemma~\ref{lem:one-sided-domination} to $H[R_j]$ and obtain a
set $S_j\subseteq R_j$ such that every vertex of $R_j\setminus S_j$ has an
in-neighbor in $S_j$. Put $A_j=A_{j-1}\cup S_j$.

The sets $S_1,\ldots,S_k$ are pairwise disjoint and each has size at most $L$.
If $v\notin A_j$, then for every $1\le \ell\le j$ the vertex $v$ belongs to
$R_\ell\setminus S_\ell$, so it has an in-neighbor in $S_\ell$. Consequently
$d_{A_j}^-(v)\ge j$. With $A=A_k$, we have $\lvert A\rvert\le kL$ and
$d_A^-(v)\ge k$ for every $v\notin A$.

Applying the same construction to the reverse digraph gives a set $B$ with
$\lvert B\rvert\le kL$ and $d_B^+(v)\ge k$ for every $v\notin B$. Let $P=A\cup B$.
Then $\lvert P\rvert\le 2kL$, and every vertex outside $P$ has at least $k$ in-neighbors
and at least $k$ out-neighbors in $P$.

For each $p\in P$, the assumption $\delta^0(H)\ge k$ allows us to choose sets
$O_p\subseteq N_H^+(p)$ and $I_p\subseteq N_H^-(p)$ with
$\lvert O_p\rvert=\lvert I_p\rvert=k$. Define
$Q=P\cup\bigcup_{p\in P}(O_p\cup I_p)$. Vertices outside $P$ already receive
the required support from $P$, while each $p\in P$ receives it from $I_p$ and
$O_p$. Finally, $\lvert Q\rvert\le(2k+1)\lvert P\rvert\le 2k(2k+1)L$. All choices are explicit and
can be made in polynomial time.
\end{proof}

\begin{lemma}[Universal coloring family]\label{lem:universal-coloring}
For all positive integers $n$ and $t$, one can deterministically construct a
finite indexed multiset
\(
\mathcal U_{n,t}=(\chi_\alpha)_{\alpha\in I},\quad
\chi_\alpha:[n]\to\{0,1\},
\)
where equal colorings may occur with different indices and are counted with
multiplicity. Uniform choice from $\mathcal U_{n,t}$ means uniform choice of an
index $\alpha\in I$. The multiset satisfies the following properties:
\begin{enumerate}
    \item for every $T\subseteq[n]$ with $\lvert T\rvert\le t$, the restrictions
    $\{\chi_\alpha|_T:\alpha\in I\}$ contain every map $T\to\{0,1\}$;
    \item $\lvert \mathcal U_{n,t}\rvert\le 16n^2\cdot 4^t$;
    \item $\mathcal U_{n,t}$ can be constructed in time polynomial in $n\,2^t$.
\end{enumerate}
\end{lemma}

\begin{proof}
Set $\varepsilon=2^{-t-1}$ and let $q=2^r$ be the smallest power of $2$
satisfying $q\ge n/\varepsilon$; thus $q<2n/\varepsilon$. We use a field $F$
of order $q$, together with a fixed basis of $F$ over $\mathbb F_2$ and a fixed
nonzero $\mathbb F_2$-linear map $\lambda:F\to\mathbb F_2$.

For completeness, such a field can be constructed deterministically in time
polynomial in $q$. Enumerate the monic binary polynomials of degree $r$ and
test each candidate for divisibility by every monic polynomial of degree $d$
with $1\le d\le\lfloor r/2\rfloor$. If a monic polynomial of degree $r$ is
reducible, then one of its nonconstant factors has degree at most
$\lfloor r/2\rfloor$. Hence a degree-$r$ candidate that passes all these tests
is irreducible. Such a candidate exists because the number of
monic irreducible binary polynomials of degree $r$ is
$r^{-1}\sum_{d\mid r}\mu(d)2^{r/d}>0$. This is immediate for $r=1,2$; for
$r\ge3$, every divisor $d>1$ gives a positive integer $r/d\le\lfloor r/2\rfloor$,
so the absolute contribution of the remaining terms is at most
$\sum_{j=1}^{\lfloor r/2\rfloor}2^j<2^{r/2+1}<2^r$. The enumeration and trial
divisions therefore take $q^{O(1)}$ time.

Let $m=\lceil n/r\rceil$. Pad each vector $x\in\mathbb F_2^n$ with zeros to
length $mr$, group its coordinates into blocks of length $r$, and use the
chosen basis to regard the blocks as elements
$a_0(x),\ldots,a_{m-1}(x)$ of $F$. This gives an injective
$\mathbb F_2$-linear map from $\mathbb F_2^n$ to $F^m$. Associate with $x$ the
polynomial $p_x(Z)=\sum_{j=0}^{m-1}a_j(x)Z^j$. If $x\ne0$, then $p_x$ is a
nonzero polynomial of degree at most $m-1$.

For each $(z,y)\in F^2$, define $\chi_{z,y}:[n]\to\{0,1\}$ by
$\chi_{z,y}(i)=\lambda(yp_{e_i}(z))$, where $e_i$ is the $i$th standard basis
vector of $\mathbb F_2^n$, and let
$\mathcal U_{n,t}=(\chi_{z,y})_{(z,y)\in F^2}$. Thus the multiset has $q^2$
indexed members, with all repetitions retained, and uniform averaging over it
is the same as uniform averaging over $(z,y)\in F^2$.

Fix a nonempty set $R\subseteq[n]$, put $x=\sum_{i\in R}e_i$, and write
$Z_x=\{z\in F:p_x(z)=0\}$. The complete averaging calculation is given in
Appendix~\ref{app:universal-coloring-calculations},
Subsection~\ref{app:small-bias-calculation}; see in particular
\eqref{eq:app-small-bias-identity}--\eqref{eq:app-small-bias-bound}. It yields
\begin{equation}\label{eq:small-bias-main}
\left\lvert \mathbb E_{\chi\in\mathcal U_{n,t}}(-1)^{\sum_{i\in R}\chi(i)}\right\rvert
\le\varepsilon.
\end{equation}

Now fix $T\subseteq[n]$, put $s=\lvert T\rvert\le t$, and fix
$\eta:T\to\{0,1\}$. The detailed Fourier expansion of the indicator of
$\chi|_T=\eta$ and the subsequent term-by-term estimate appear in
Appendix~\ref{app:universal-coloring-calculations},
Subsection~\ref{app:extension-probability-calculation}; see
\eqref{eq:app-extension-probability}. Using \eqref{eq:small-bias-main}, that
calculation gives
\begin{equation}\label{eq:extension-probability-main}
\Pr_{\chi\in\mathcal U_{n,t}}(\chi|_T=\eta)>0.
\end{equation}
Hence at least one member of the indexed multiset extends $\eta$, proving the
universal-coloring property.

Finally, $\lvert \mathcal U_{n,t}\rvert=q^2<4n^2/\varepsilon^2=16n^2\cdot 4^t$. Field
operations, polynomial evaluations, and the generation of all $q^2$
colorings take time polynomial in $q$ and $n$; since $q<2^{t+2}n$, this is
polynomial in $n\,2^t$.
\end{proof}

\medskip

\begin{algorithm}[H]
\caption{Deciding a $(\delta^0\ge k_1,\delta^0\ge k_2)$-partition}
\label{alg:zero-zero-partition}
\noindent\textbf{Input:} A semicomplete digraph $D$ given by its adjacency matrix, and fixed integers $k_1,k_2\ge2$.\par
\noindent\textbf{Output:} A $(\delta^0\ge k_1,\delta^0\ge k_2)$-partition of $D$, or a declaration that no such partition exists.
\begin{enumerate}[label=\textup{(\arabic*)},leftmargin=*,itemsep=2pt]
    \item Let $D$ have $n$ vertices, fix a bijection between $V(D)$ and $[n]$, and fix a total order on $V(D)$ for deterministic tie-breaking.
    \item Put $c(k):=2k(2k+1)$, $L:=\lceil\log(n+1)\rceil$, and $t:=(c(k_1)+c(k_2))L$; the support bound used to choose $t$ comes from Lemmas~\ref{lem:one-sided-domination} and~\ref{lem:small-support}.
    \item Construct $\mathcal U_{n,t}$ by Lemma~\ref{lem:universal-coloring}, identifying its colors $0,1$ with $1,2$, respectively.
    \item For each initial coloring $\chi:V(D)\to\{1,2\}$ in this family, set $W_i:=\{v:\chi(v)=i\}$. A vertex $v\in W_i$ is \emph{bad} if $d_{W_i}^+(v)<k_i$ or $d_{W_i}^-(v)<k_i$.
    \item At the initial state and after every recoloring, first test whether $W_1$ and $W_2$ are both nonempty and satisfy $\delta^0(D[W_i])\ge k_i$ for $i=1,2$; if so, output this partition.
    \item If the test fails and either $n$ recoloring steps have already been performed in the current run or there is no bad vertex, terminate that run and continue with the next initial coloring.
    \item Otherwise choose the first bad vertex in the fixed total order, change its color from $i$ to $3-i$, and repeat the current run from the test step.
    \item If all initial colorings are exhausted without acceptance, answer that no required partition exists.
\end{enumerate}
\end{algorithm}

\medskip

\begin{proof}[Proof of Theorem~\ref{thm:fixed-semidegree-partition}]
\emph{Soundness.}
The algorithm accepts only after directly verifying the two semidegree
conditions and the nonemptiness of both parts, so every output partition is
valid.

\emph{Completeness.}
Suppose that a valid partition
$V(D)=V_1^*\cup V_2^*$ exists, and define the target color
$\tau(v)=i$ for $v\in V_i^*$. Apply Lemma~\ref{lem:small-support} to each
$D[V_i^*]$. This gives sets $Q_i\subseteq V_i^*$ of size at most
$c(k_i)\lceil\log(\lvert V_i^*\rvert+1)\rceil\le c(k_i)L$ such that every vertex of
$V_i^*$ has at least $k_i$ in-neighbors and at least $k_i$ out-neighbors in
$Q_i$. Since $k_i\ge1$ and $V_i^*$ is nonempty, this also implies
$Q_i\ne\varnothing$. Hence $\lvert Q_1\cup Q_2\rvert\le t$, and
Lemma~\ref{lem:universal-coloring} supplies an initial coloring $\chi_0$
satisfying $\chi_0(q)=i$ for every $q\in Q_i$ and $i\in\{1,2\}$.

Consider the run beginning with $\chi_0$. We claim inductively that every
vertex of $Q_i$ retains color $i$ throughout the run. The claim holds
initially. If it holds before a recoloring, then $Q_i\subseteq W_i$, so every
vertex $v\in V_i^*$ whose current color is $i$ satisfies
$d_{W_i}^+(v)\ge d_{Q_i}^+(v)\ge k_i$ and
$d_{W_i}^-(v)\ge d_{Q_i}^-(v)\ge k_i$. Such a vertex is not bad. In
particular, no vertex of $Q_1\cup Q_2$ is recolored, and the invariant
persists.

It follows that every selected bad vertex has a current color different from
its target color. Since there are only two colors, recoloring that vertex
changes it precisely to its target color. Thus the potential
$\Phi(\chi)=\lvert \{v\in V(D):\chi(v)\ne\tau(v)\}\rvert$ decreases by exactly one at
every recoloring, and no correctly colored vertex is ever recolored.
Consequently this run performs at most $n$ recolorings before reaching the
target partition. Moreover, the invariant keeps $Q_i\subseteq W_i$ for both
$i$, and each $Q_i$ is nonempty, so neither color class becomes empty. If at
some earlier point there is no bad vertex, the acceptance test succeeds rather
than terminating the run. Therefore the algorithm necessarily accepts this
run.

\emph{Running time.}
Let $C=c(k_1)+c(k_2)$. Since $t=CL$ and $2^L<2(n+1)$, we have
$4^t=(2^L)^{2C}<(2(n+1))^{2C}$. Lemma~\ref{lem:universal-coloring} therefore
produces at most $16n^2\cdot 4^t=n^{O_{k_1,k_2}(1)}$ initial colorings. Each run
performs at most $n$ recolorings, and a direct recomputation of all relevant
in- and outdegrees takes $O(n^2)$ time per state. The total running time is
$n^{O_{k_1,k_2}(1)}$, which is polynomial because $k_1$ and $k_2$ are fixed.
\end{proof}

\section{Remarks}\label{section 4}

The three algorithms share a bounded-certificate strategy.  Semicompleteness bounds the number of vertices of small indegree or outdegree, so a feasible partition can be encoded by a small exceptional set, a bounded core, or a logarithmic support set.  In the $(\delta^+\geq k_1,\delta^-\geq k_2)$ case, an extremal partition and a closure operation recover one part exactly.  In the $(\delta^+\geq k_1,\delta^0\geq k_2)$ case, a protective set and the maximum out-core allow controlled vertex transfers.  In the $(\delta^0\geq k_1,\delta^0\geq k_2)$ case, a universal coloring family contains a coloring that agrees with the target colors on two support sets, after which deterministic recoloring reaches a valid partition.  Thus, structural properties of semicomplete digraphs are converted into explicit certificates that can be recovered in polynomial time, resolving the three open cases constructively.

The polynomial-time bounds require $k_1$ and $k_2$ to be fixed. If they are part of the input, the same algorithms yield only $\mathrm{XP}$ running times parameterized by $k_1+k_2$. Here, an algorithm is XP if, for every fixed parameter $k$, it runs in time $n^{f(k)}$ for some computable function $f$.

\begin{problem}
Let $D$ be a semicomplete digraph, and let integers $k_1,k_2\geq 2$ be part of the input.  Determine the parameterized complexity, with parameter $k_1+k_2$, of deciding whether $D$ has
\begin{itemize}
    \item a $(\delta^+\geq k_1,\delta^-\geq k_2)$-partition;
    \item a $(\delta^+\geq k_1,\delta^0\geq k_2)$-partition;
    \item a $(\delta^0\geq k_1,\delta^0\geq k_2)$-partition.
\end{itemize}
In particular, is each problem fixed-parameter tractable, or is at least one of them $\mathrm{W[1]}$-hard?
\end{problem}

\medskip

\textbf{Declaration on the use of generative AI}

The authors used generative AI tools { (ChatGPT 5.5 Pro and 5.5 Thinking)} to assist in numerical computation, checking proofs and improving exposition.

\appendix
\section{Detailed calculations for Lemma~\ref{lem:universal-coloring}}
\label{app:universal-coloring-calculations}

This appendix records the two calculations used in the proof of Lemma~\ref{lem:universal-coloring}. All expectations and probabilities below are taken uniformly over the indexed multiset $\mathcal U_{n,t}=(\chi_{z,y})_{(z,y)\in F^2}$, so repetitions are counted with their multiplicities.

\subsection{The small-bias average}
\label{app:small-bias-calculation}

Fix a nonempty set $R\subseteq[n]$ and put $x=\sum_{i\in R}e_i\in\mathbb F_2^n$. Since $R$ is nonempty, $x\ne0$, and hence $p_x$ is a nonzero polynomial of degree at most $m-1$. The maps $x\mapsto p_x$ and $\lambda$ are $\mathbb F_2$-linear, so for every $(z,y)\in F^2$ we have, in $\mathbb F_2$,
\[
\sum_{i\in R}\chi_{z,y}(i)=\sum_{i\in R}\lambda\bigl(yp_{e_i}(z)\bigr)=\lambda\left(y\sum_{i\in R}p_{e_i}(z)\right)=\lambda\left(yp_{\sum_{i\in R}e_i}(z)\right)=\lambda\bigl(yp_x(z)\bigr).
\]
Consequently,
\begin{equation}
\mathbb E_{\chi\in\mathcal U_{n,t}}(-1)^{\sum_{i\in R}\chi(i)}=\frac1{q^2}\sum_{z\in F}\sum_{y\in F}(-1)^{\sum_{i\in R}\chi_{z,y}(i)}=\frac1{q^2}\sum_{z\in F}\sum_{y\in F}(-1)^{\lambda(yp_x(z))}.
\label{eq:app-small-bias-double-sum}
\end{equation}

For fixed $z\in F$, set $S_z=\sum_{y\in F}(-1)^{\lambda(yp_x(z))}$. If $p_x(z)=0$, then every summand equals $1$, and therefore $S_z=q$. If $p_x(z)\ne0$, multiplication by $p_x(z)$ is a bijective $\mathbb F_2$-linear map on $F$, so $L_z(y)=\lambda(yp_x(z))$ is a nonzero $\mathbb F_2$-linear functional. Its kernel has dimension $r-1$ and therefore contains $q/2$ elements. Thus $L_z$ takes each value in $\mathbb F_2$ exactly $q/2$ times, and
\[
S_z=\frac q2(-1)^0+\frac q2(-1)^1=\frac q2-\frac q2=0.
\]

Writing $Z_x=\{z\in F:p_x(z)=0\}$ and substituting these two values of $S_z$ into \eqref{eq:app-small-bias-double-sum}, we obtain
\begin{equation}
\mathbb E_{\chi\in\mathcal U_{n,t}}(-1)^{\sum_{i\in R}\chi(i)}=\frac1{q^2}\sum_{z\in F}S_z=\frac1{q^2}\left(\sum_{z\in Z_x}q+\sum_{z\in F\setminus Z_x}0\right)=\frac{\lvert Z_x\rvert}{q}.
\label{eq:app-small-bias-identity}
\end{equation}
Because $p_x$ is nonzero and has degree at most $m-1$, it has at most $m-1$ roots. Moreover, $m=\lceil n/r\rceil$ with $r\ge1$, so $m-1\le n$, while the choice of $q$ gives $n/q\le\varepsilon$. Therefore
\begin{equation}
\left\lvert \mathbb E_{\chi\in\mathcal U_{n,t}}(-1)^{\sum_{i\in R}\chi(i)}\right\rvert=\frac{\lvert Z_x\rvert}{q}\le\frac{m-1}{q}\le\frac nq\le\varepsilon.
\label{eq:app-small-bias-bound}
\end{equation}

\subsection{The probability of extending a prescribed coloring}
\label{app:extension-probability-calculation}

Fix $T\subseteq[n]$, let $s=\lvert T\rvert\le t$, and fix a map $\eta:T\to\{0,1\}$. For each $i\in T$, the indicator of the equality $\chi(i)=\eta(i)$ is $\mathbf 1_{\{\chi(i)=\eta(i)\}}=(1+(-1)^{\chi(i)+\eta(i)})/2$. Multiplying these indicators and expanding the product over the choices of the nonconstant term gives the Walsh--Fourier expansion
\begin{align}
\mathbf 1_{\{\chi|_T=\eta\}}
&=\prod_{i\in T}\mathbf 1_{\{\chi(i)=\eta(i)\}}=2^{-s}\prod_{i\in T}\left(1+(-1)^{\chi(i)+\eta(i)}\right)\notag\\
&=2^{-s}\sum_{R\subseteq T}\prod_{i\in R}(-1)^{\chi(i)+\eta(i)}=2^{-s}\sum_{R\subseteq T}(-1)^{\sum_{i\in R}\eta(i)}(-1)^{\sum_{i\in R}\chi(i)}.
\label{eq:app-indicator-fourier}
\end{align}
Taking expectations in \eqref{eq:app-indicator-fourier} and using $\mathbb E[\mathbf 1_A]=\Pr(A)$ yields
\begin{align}
\Pr_{\chi\in\mathcal U_{n,t}}(\chi|_T=\eta)
&=2^{-s}\sum_{R\subseteq T}(-1)^{\sum_{i\in R}\eta(i)}\mathbb E_{\chi\in\mathcal U_{n,t}}(-1)^{\sum_{i\in R}\chi(i)}\notag\\
&=2^{-s}\left(1+\sum_{\varnothing\ne R\subseteq T}(-1)^{\sum_{i\in R}\eta(i)}\mathbb E_{\chi\in\mathcal U_{n,t}}(-1)^{\sum_{i\in R}\chi(i)}\right).
\label{eq:app-probability-before-bound}
\end{align}
The term $R=\varnothing$ contributes $1$. Applying the triangle inequality to the remaining terms and then using \eqref{eq:app-small-bias-bound} for every nonempty $R\subseteq T$ gives
\begin{align}
\Pr_{\chi\in\mathcal U_{n,t}}(\chi|_T=\eta)
&\ge 2^{-s}\left(1-\sum_{\varnothing\ne R\subseteq T}\left\lvert \mathbb E_{\chi\in\mathcal U_{n,t}}(-1)^{\sum_{i\in R}\chi(i)}\right\rvert\right)\notag\\
&\ge 2^{-s}\left(1-\sum_{\varnothing\ne R\subseteq T}\varepsilon\right)=2^{-s}\bigl(1-(2^s-1)\varepsilon\bigr).
\label{eq:app-probability-lower-bound}
\end{align}
There are $2^s-1$ nonempty subsets of $T$. Since $s\le t$ and $\varepsilon=2^{-t-1}$,
\[
(2^s-1)\varepsilon\le(2^t-1)2^{-t-1}=\frac12-2^{-t-1}<\frac12.
\]
Combining this with \eqref{eq:app-probability-lower-bound} gives
\begin{equation}
\Pr_{\chi\in\mathcal U_{n,t}}(\chi|_T=\eta)\ge 2^{-s}\bigl(1-(2^s-1)\varepsilon\bigr)>2^{-s-1}>0.
\label{eq:app-extension-probability}
\end{equation}

\end{spacing}
\end{document}